\documentclass[11pt,reqno]{amsart}
\usepackage{amsmath,amssymb}
\usepackage{amsthm} 
\usepackage{amscd}

\usepackage[pdftex]{graphicx}
\usepackage[all]{xy}

\makeatletter
\@addtoreset{equation}{section}
\makeatother
\usepackage{enumerate}

\def\CC{{\mathbb C}}

\def\RR{{\mathbb R}}

\def\PP{{\mathbb P}}

\def\ee{{\mathrm e}}

\def\ii{{\sqrt{-1}}}

\def\cE{{\mathcal E}}

\def\cK{{\mathcal K}}
\def\cL{{\mathcal L}}
\def\cM{{\mathcal M}}

\def\cZ{{\mathcal Z}}

\def\fb{{\mathfrak b}}

\def\ft{{\mathfrak t}}

\def\rU{{\mathrm U}}

\def\ri{{\mathrm i}}
\def\rr{{\mathrm r}}

\def\tA{{\widetilde A}}

\def\tK{{\widetilde K}}

\def\tgamma{{\widetilde \gamma}}

\def\hJ{{\widehat J}}

\def\hX{{\widehat X}}

\def\hvarpi{{\widehat \varpi}}

\def\qed{\hbox{\vrule height6pt width3pt depth0pt}}

\def\nuI#1{{\nu^{\mathrm{I}}_{#1}}}

\def\book#1{\rm{#1}, }
\def\paper#1{\textit{#1}, }
\def\jour#1{\rm{#1}, }
\def\yr#1{({\rm{#1}) }}
\def\vol#1{\textbf{#1}}
\def\pages#1{\rm{#1}}

\def\publaddr#1{\rm{#1}, }
\def\publ#1{\rm{#1}, }
\def\by#1{{\rm{#1}, }}

\newtheorem{theorem}{Theorem}[section]

\newtheorem{lemma}[theorem]{Lemma}

\def\book#1{\rm{#1}, }
\def\paper#1{\textit{#1}, }
\def\jour#1{\rm{#1}, }
\def\yr#1{({\rm{#1}) }}
\def\vol#1{\textbf{#1}}
\def\pages#1{\rm{#1}}

\def\publaddr#1{\rm{#1}, }
\def\publ#1{\rm{#1}, }
\def\by#1{{\rm{#1}, }}

\begin{document}


\title{Statistical mechanics of elastica for the shape of supercoiled DNA: 
hyperelliptic elastica of genus three}

\author{Shigeki Matsutani}
%

\date{\today}

\begin{abstract}
This article studies the statistical mechanics of elastica as a model of the shapes of the supercoiled DNA, and shows that its excited states can be characterized by the focusing modified KdV (MKdV) equation due to thermal fluctuation.
Following the previous paper (Matsutani and Previato, Physica D {\bf{430}} (2022) 133073), the hyperelliptic solutions of the focusing modified KdV (MKdV) equation of genus three are considered.
There appears a pattern as a repetition of the modulation of figure-eight and the inverse {\lq}S{\rq} as a thermal fluctuation of elastica, called the S-eight mode.
Our model states that the excited states of elastica due to the thermal effect have the S-eight mode, which reproduces the shapes of the AFM image of the supercoiled DNAs observed by Japaridze et al. (Nano Lett. \vol{17} 3, \yr{2017} 1938).
\end{abstract}

\maketitle

\section{Introduction}\label{sec:1}
In the previous paper \cite{MP22}, the author with Emma Previato investigated an algebro-geometric model for the shape of supercoiled DNA.
As mentioned there, the mathematical description of the shape of the supercoiled DNA is a challenging problem in which no one can find the shape mathematically.
Since the shape of the supercoiled DNA plays crucial roles in life \cite{Betal, CDLT, LS, Petal, VV}, there are many studies on the shape \cite{BM, GPL, KP, LS, SCT, TsuruWadati}.
The electron microscope images of DNA on a plane show that, except in a few cases, the shapes of the loop are much more complicated than Euler's elastica, a circle and eight-figure.
Furthermore, it is neither squeezed nor dense, but is characterized by voids between intersections that are weakly governed by elastic forces.
These properties mean that it cannot be realized as a minimal state of its Euler-Bernoulli energy functional even by considering its three-dimensional effect;
 the voids cannot appear mathematically as long as we consider the minimal state of a given energy functional.
The minimal state cannot have any other parameters and is therefore expressed by the elliptic functions, which have only double periods and no ability to express complicated shapes.

In other words, we should move on to consider the excited state of supercoiled DNA due to thermal effect, i.e., the statistical mechanics of geometric objects.
The author sometimes called this model quantized elastica due to the analogy between the Planck constant $\hbar$ and the inverse of the temperature $\beta$, since the statistical mechanics of geometric objects corresponds to the quantization of geometric objects after Wick rotation; 
(Two-dimensional quantum gravity is related to the statistical mechanics of random triangulation of some Riemann surfaces \cite{Itzykson,Witten}.
The statistical mechanics of the liquid crystal observed in the laboratory is related to the $O(3)$ gauge field theory \cite{Frenkel}.)

Since the elastica model has the constraint that it does not stretch, it cannot be regarded as an extension of the random walk model \cite{DoiEdwards}, but it has an interesting internal symmetry related to the hyperelliptic curves.
The author proposed a model of the statistical mechanics of elastica to express the shapes of the supercoiled DNA in 1998 \cite{Mat97}.
The shapes can only be realized if thermal effects are taken into account, and must be the excited states of the elastica rather than minimal states. 
The excited states of the elastica on the plane are well described by curves whose tangential angles $\phi$ obey the modified KdV (MKdV) equation \cite{AS},
\begin{equation}
(\partial_{t}+\alpha \partial_s)\phi
           +\frac{1}{8}
\left(\partial_{s} \phi \right)^3
+\frac{1}{4}\partial_{s}^3 \phi=0,
\label{4eq:rMKdV2}
\end{equation}
where $t$ and $s$ are the real axes and $\alpha$ is a real parameter.
Here $t$ does not mean the physical time axis, but one of the inner-space directions of the excited states due to thermal fluctuations.
(The auxiliary time $t$ corresponds to the Schwinger time of the quantum field theory \cite{Ramond, Schwinger} in this system, which shows the propagation from an excited state to the others. )
It is shown that the inner symmetry of the statistical mechanics of elastica is governed by the hyperelliptic curves and their moduli whereas the symmetry of the classical elastica at zero-temperature is determined by the elliptic curves and their moduli \cite{Mat97, Mat02b, MO03a, MP16}.
Thus we have also referred to this model as generalized elasticae \cite{MP16}.

The statistical mechanics of elastica is a kind of solvable model of statistical mechanics, which can be expressed by solutions of the integrable system, the MKdV equation.
Finding the solutions, however, requires advanced knowledge of modern algebraic geometry and the computational techniques of hyperelliptic functions.
Therefore, no one has attacked this problem.

The author and Emma Previato decided to solve the problem of elastica in 2004 based on the papers \cite{Mat02b, Pr93}.
The author has developed the tools together with co-authors \cite{Mat02b, MP16, MP22, EEMOP08, EMO08} following Baker's approach and a similar movement \cite{Baker97, BEL97b, BuL04, BEL20}.
To solve the problem, they considered that a novel approach was needed that directly connects the algebraic curves, rather than the theta-function approach \cite{P0}.
They have spent two decades refining and reconstructing Abelian function theory, including hyperelliptic function theory, as problems in algebraic geometry \cite{Mat10, MP22, M24}.
(The computational techniques of hyperelliptic functions have the ability to describe some complicated systems that cannot be described by elliptic functions. 
We note that the development of computational techniques provided solutions to the problem of the motion of a particle around the black hole \cite{EHKKLS}).

Since supercoiled DNA exists in three-dimensional space, its statistical mechanical model of excited states of elastica is governed by the nonlinear Schr\"odinger (NLS) equation and the complex MKdV (CMKdV) equation rather than the MKdV equation \cite{Mat99a}.
However, since the observed shapes of supercoiled DNA are on a plane and the solutions of the NLS and CMKdV equations are much more complicated than the MKdV equation, we focus on the two-dimensional model.
Though the hyperelliptic solutions have not been concretely obtained, the elliptic function solutions of (two-dimensional) elastica, figure-eight due to Euler \cite{Euler44}, have been studied well which correspond to the ground state of the model, and are often observed in AFM \cite{Petal}.

Thus the author attacked this two-dimensional problem in the previous paper with Previato \cite{MP22} to find the real analytic solution of the MKdV equation in terms of the data of hyperelliptic curves of genus two \cite{Mat02b}.
Then the attempt failed in \cite{MP22}.
It was clarified that any hyperelliptic non-degenerate curves of genus two could not provide the excited states of elastica because of the reality conditions \cite{MP22}.
As a conclusion of \cite{MP22}, hyperelliptic curves of higher genus ($g \ge 3$) should be investigated to find the solution of (\ref{4eq:rMKdV2}).
Therefore, the associated paper \cite{M24a} is devoted to finding the solutions of (\ref{4eq:rMKdV2}) in terms of the meromorphic functions of hyperelliptic curves of genus three.
There we find some approximate solutions of (\ref{4eq:rMKdV2}) of the gauged MKdV equation.

This paper, based on the results in the associated paper \cite{M24a}, shows that some hyperelliptic curves of genus three supply the forms of the excited states of elastica which have never been obtained.
They exhibit typical shapes, modulation of a repetition of figure-eight and inverse of {\lq}S\rq.
We call it the S-eight mode.
In other words, our model states that the excited states of the statistical mechanics of elastica have the S-eight mode.
It is quite surprising that we find a similar shape in some of the AFM images of a supercoiled DNA in \cite[Figure 4]{JMB}. 
The observed shapes of supercoiled DNA have the S-eight mode.

The figure-eight given by Euler in 1744 which appears similar shapes of the short closed supercoiled DNAs, e.g. in \cite{Petal} but no one has ever mathematically reproduce any more complicated shape of supercoiled DNA with voids.
We emphasize that this demonstration shows that an important first step has been made toward the complete mathematical expression of the supercoiled DNAs.

The content is following:
Section 2 reviews the statistical mechanics of elastica as in \cite{Mat97}.
 previous results, which is the same as \cite{MP22}.
Section 3 is devoted to the geometry of the hyperelliptic curves and their relation to the MKdV equation over $\CC$ and the gauge MKdV equation.
Section 4 provides hyperelliptic solutions of the gauged MKdV equation of genus three.
Section 5 shows the numerical evaluation of the solutions and provides the discussion of the solutions.
We consider the computational results and their relations to the supercoiled DNA. Section 6 gives the conclusion of this paper.

\section{Statistical Mechanics of elastica}
\label{sec:SMelas}

We consider that the supercoiled DNA is caused by thermal fluctuations.
As a model, we could assume that the stretching must be suppressed and obey the elastic force known as the Euler-Bernoulli energy function.
When supercoiled DNAs are observed, e.g. by atomic force microscopy, they lie on a plane.
Therefore, we consider their shape of the flat plane. 
In other words, we consider the statistical mechanics of elastica on a plane.

Let us consider an analytic and isometric immersion $Z$ of a loop $S^1$ into the complex plane $\CC$ \cite{MP16},
$$
\cM
:=\{ Z: S^1 \hookrightarrow \CC\ | \ \mbox{analytic, isometric} \}.
$$
Let $ds^2 =\overline{d Z} dZ$ be the differential of arclength $s$ such that $\displaystyle{\int ds = 2\pi}$.
Then its tangential vector is given by $\partial_s Z$ satisfying $|\partial_s Z|=1$ and its curvature is $\displaystyle{k:=\frac{1}{\ii}\partial_s\log \partial_s Z}$ for $\partial_s := \partial/\partial s$.
Let $\partial_s Z=\ee^{\ii \psi}$ where $\psi$ is a real value and $k = \partial_s \psi$.
The elastica is given as the minimal point of the Euler-Bernoulli energy
\begin{equation}
 \cE[Z]:=\frac{1}{2}\oint_{S^1} \, k^2(s) d s 
= \oint_{S^1} \{Z, s\}_{\mathrm{SD}} ds,
\label{4eq:EBenergy}
\end{equation}
where $\displaystyle{\{Z, s\}_{\mathrm{SD}}={\small{
\frac{\partial }{\partial s}
\left(\frac{\partial^2 Z/\partial s^2}{\partial Z/\partial s}\right)
-\frac12\left(\frac{\partial^2 Z/\partial s^2}{\partial Z/\partial s}\right)^2}}}$ is the Schwarzian derivative.

Let us consider the partition function $\cZ[\beta]$ of the statistical mechanics of elastica, which is given by
\begin{gather}
\cZ[\beta]= \int_\cM DZ \exp(-\beta \cE[Z])
\end{gather}
$DZ$ is the Feynman measure of this system, and $1/\beta$ is the temperature.
Due to $\cM$ we only consider the isometric deformation.
We remark that this isometric condition is strong and this model is completely different from the well-known random chain model \cite{DoiEdwards}.

To evaluate this partition function $\cZ[\beta]$, we consider the infinitesimal deformation $Z \to Z+\delta Z$ parameterized by the infinite dimensional virtual time parameter $\ft = (t_1, t_2, \ldots)$.
We apply the effective action method to this evaluation 
(The time $\ft$ corresponds to the Schwinger time of the quantum field theory \cite{Ramond, Schwinger} in this system, which shows the propagation from one excited state to another).
In other words, $\delta Z=\partial_t Z\delta t=(u_1(s)+\ii u_2(s))\partial_s Z \delta t$ must preserve the isomorphic deformation, i.e. $[\partial_t, \partial_s]=0$.
By the Goldstein-Petrich condition $[\partial_t, \partial_s]=0$ \cite{GoldsteinPetrich1}, we have
$\partial_t \partial_s Z=\ii \partial_t \psi \partial_s Z$, while
$\partial_s \partial_t Z=[(\partial_s u_1+k u_2)+\ii(\partial_s u_2 - k u_1)] \partial_s Z$.
The isometric deformation is given by
$$
\partial_t k = \Omega u_2, \quad u_1=\partial_s^{-1}k u_2,\quad
\Omega:=\partial_s^2 + \partial_s k \partial_s^{-1}k .
$$
Following the theory of the effective action for the functional integral method \cite{Ramond}, we consider the deformation such that the Euler-Bernoulli energy changes as $E[Z+\delta Z]=E[Z]+\displaystyle{ \delta t\int k (\partial_t k) d s}+ O(\delta t^2)$.
We require the condition that $\delta E:=\displaystyle{\delta t\int k (\partial_t k) d s}=0$.

The condition $|\partial_s Z|=1$, namely $\partial_s Z \in \rU(1)$, allows us to define an action $g_{s_0}Z(s) = Z(s- s_0)$ of $\rU(1)$.
It can be expressed as $(\partial_s - \partial_t)\partial_s Z=0$ and $(\partial_s - \partial_t)\partial_s k=0$.
This deformation can be regarded as the trivial isometric deformation and the isoenergy deformation, i.e. the trivial gauge transformation.
Thus we introduce the infinite deformations with infinite proper time $\ft =(t_1, t_2, \cdots)$ by starting the trivial transformation as the initial deformation,
$$
\partial_{t_{n+1}}k = \Omega \partial_{t_{n}}k,\quad
\partial_{t_1} k = \partial_s k.
$$
Then $\partial_{t_n} k = \Omega^n \partial_s k$ is consistent with the $n$-th focusing modified KdV (MKdV) equation \cite{MP16}.
For the deformation, the Euler-Bernoulli energy is preserved, i.e., $E[Z+\delta Z]=E[Z]+ \delta t_n\displaystyle{\int k (\partial_{t_n} k) d s}= E[Z]$.
The requirement that the vanishing of the first-order deformation $\delta E=\delta t\displaystyle{\int k (\partial_t k ) d s}=0$ guarantees complete energy preservation due to the integrability of the MKdV equation, i.e., $\partial_{t_n} E[Z]=0$.
In other words, the statistical mechanics of elastica has infinite gauge symmetry, in general.
Since the MKdV equation preserves the Euler-Bernoulli energy of the deformation of curves. $\cM$ is decomposed to
$$
\cM = \bigsqcup_{E\in \RR} \cM_{E},
\quad \cM_{E}:= \{ Z\in \cM\ |\ \cE[Z]=E\},
$$
which naturally includes the Euler's elastica as its minimal point (the state of the zero temperature.)

It has the orbital decomposition in terms of the loop soliton whose curvature obeys the modified KdV equation;

Then $\cZ[\beta]$ can be evaluated like
$$
\cZ[\beta] = \int d E \frac{\partial \mathrm{Vol}(\cM_E)}{\partial E}
 \ee^{-\beta E},
$$
where $\mathrm{Vol}(\cM_E)$ is the volume of $\cM_E$.

Thus to evaluate the statistical mechanics of the elastica is reduced to finding the solutions of the modified KdV equations.

The evaluation of the statistical mechanics of elastica is given as follows.
\begin{enumerate}

\item[(0)]
to regard the elastica model as the minimal point or the state at the zero temperature.

\item to find the solutions of the (focusing) MKdV equation as the excited state of the Euler-Bernoulli energy associated with the hyperelliptic curves $X$ of genus $g$,

\item to find the energy-preserving orbit consisting of the iso-energy state $\cM_E$, which is a real subspace $\cM_{X,E}$ of the Jacobi varieties $J_X$ of $X$ (given by real period matrices),

\item to evaluate the moduli dependence of the volume of $\cM_{X,E}$ for the moduli parameters of $X$ as subspace of the moduli space of the hyperelliptic curves.

\item to assign the topology in $\cM$ and to define a suitable measure with respect to the Euler-Bernoulli energy, and

\item to analyze $\cM$ for a given $\beta$.

\end{enumerate}

The author with Previato investigated the geometrical structure of $\cM$ as a generalized elastica problem \cite{Mat10, MO03a, MP16, Pr93}.
However finding the solutions of the focusing MKdV equation is quite difficult.
The reality conditions CI, CII, CIII as mentioned after Theorem \ref{4th:MKdVloop} restrict the space $\cM$.
The infinite symmetry can regarded as the loop group action \cite{AdamsHarnadPreviato} but reality condition restricts the group action, which quite differs from the loop group over $\CC$.

There is a mystery in the research of the shapes of supercoiled DNA, that though we have found the circle types, the eight-figure types, and much complicated shapes with several holes of supercoiled DNA in the AFM observations, we cannot find twisted loops on a plane with two or three holes.
As in \cite{MP22}, the real analytic solution of the focusing MKdV equation can be found for the cases corresponding to the circle and the eight-figure of the elliptic curve $g=1$ but cannot found for the hyperelliptic curves of genus two due to the reality condition.
The real hyperelliptic solution of genus two might be related to twisted loops with two or three holes \cite{MP22} if exists.
We require higher genus curves to obtain the real analytic solution of the focusing MKdV equation.
Since the higher twisted loops observed in the laboratory are much more complicated than the figure eight, we conjecture that this phenomenon must be related to the difficulty of the existence of the real analytic solution of the focusing MKdV equations.

The shape of the supercoiled DNA might be determined by profound reality of the algebraic geometry over the complex field.
We believe that there is a deep relationship between nature (life sciences) and mathematics, and that symmetry is crucial.

\bigskip

If we replace $\beta$ with $\ii/\hslash$, the problem is reduced to a functional integration representation of a quantum system related to a geometrical object. In theoretical physics, the method to sum a weight with respect to an invariant functional over possible states is known as the Feynman path integral (functional integral) method for statistical mechanics and quantum mechanics. 
The statistical mechanics of elastica \cite{Mat97} can be seen as a problem in which we quantize geometric objects with a non-trivial (classical or ground state) equation, i.e. the static MKdV equation, and a non-trivial constraint, i.e. isometry. 
Thus the author called it quantized elastica \cite{Mat02b}.
A similar idea appeared in \cite{Brylinski}, but it was too difficult to solve in the 1990s \cite{Mat97}.
However, the patient efforts of the last three decades to develop the study of Abelian function theory \cite{Mat02b, MP16, MP22, EEMOP08, EMO08, BEL97b, BuL04, BEL20, M24} have allowed a concrete attack on the solution.
For example, by Buchstaber Leykin \cite{BuL04}, we can study the moduli space of hyperelliptic curves by using the heat equation on the moduli parameter $b_i$ in (\ref{4eq:hypC}).
In other words, due to the recent development, this model could be explicitly evaluated using these mathematical tools and could be compared with the experimental results in the laboratory \cite{MP22}.
Thus, this model itself is quite crucial and interesting for the life science and theoretical physics.

Hence the solution of the MKdV equation can be regarded as an excited state of elastica due to thermal effect.
We investigate the explicit solutions using the hyperelliptic function, while the ground state of the elastica (at zero temperature) is expressed by the elliptic function.
Though in the above formulation we consider a loop (closed curve), in this paper we consider curves whose curvature obeys the MKdV equation.

\section{Hyperelliptic solutions of the focusing modified KdV equation over $\CC$}
\label{sec:HESGE}

We review the solution of the generalized elastica problem of excited states of elastica \cite{Mat02b} for a hyperelliptic curve $X_g$ of genus $g$ over $\CC$,
\begin{equation}
X_g:=\left\{(x,y) \in \CC^2 \ |
\ y^2 = (x-b_1)(x-b_2)\cdots(x-b_{2g+1})\right\}
\cup \{\infty\},
\label{4eq:hypC}
\end{equation}
where $b_i$'s are mutually distinct complex numbers.
Let $\lambda_{2g}=\displaystyle{-\sum_{i=1}^{2g+1} b_i}$ and $S^k X_g$ be the $k$-th symmetric product of the curve $X_g$. 
The Abelian integral $v : S^k X_g \to \CC^g$, $(k=1, \ldots, g)$ is defined by its $i$-th component $v_i$ $(i =1, \ldots, g)$,
\begin{equation}
v_i((x_1,y_1),\ldots,(x_k,y_k))=\sum_{j=1}^k
 v_i(x_j,y_j), 
\quad
v_i(x,y) = \int^{(x, y)}_\infty \nuI{i},\quad
\nuI{i} = \frac{x^{i-1}d x}{2y}.
\label{4eq:firstdiff}
\end{equation}

\cite{Mat02b} shows the hyperelliptic solutions of the MKdV equation over $\CC$,
\begin{theorem} {\textrm{\cite{Mat02b}}}
\label{4th:MKdVloop}
For $((x_1,y_1),\cdots,(x_g,y_g)) \in S^g X_g$, a fixed branch point $b_a$ $(a=1, 2, \ldots, 2g+1)$, and $u:= v( (x_1,y_1),$ $\cdots,(x_g,y_g) )$,
$$
\displaystyle{
   \psi(u) :=
-\ii \log (b_a-x_1)(b_a-x_2)\cdots(b_a-x_g)
}
$$
satisfies the MKdV equation over $\CC$,
\begin{equation}
	(\partial_{u_{g-1}}-\frac{1}{2}
(\lambda_{2g}+3b_a)
          \partial_{u_{g}})\psi
           -\frac{1}{8}
\left(\partial_{u_g} \psi\right)^3
 -\frac{1}{4}\partial_{u_g}^3 \psi=0,
\label{4eq:loopMKdV2}
\end{equation}
where $\partial_{u_i}:= \partial/\partial u_i$ as an differential identity in $S^g X_g$ and $\CC^g$.
\end{theorem}

We, here, emphasize the difference between the MKdV equations (\ref{4eq:rMKdV2}) over $\RR$ and (\ref{4eq:loopMKdV2}) over $\CC$.
The difference is crucial since we want to obtain solutions of (\ref{4eq:rMKdV2}), not (\ref{4eq:loopMKdV2}).
However, the latter is expressed well in terms of the hyperelliptic function theory.
In the associated paper \cite{M24a}, we have investigated the solutions of (\ref{4eq:rMKdV2}) based on the solutions of (\ref{4eq:loopMKdV2}).

As mentioned in \cite[(11)]{MP22}, we describe the difference.
By introducing real and imaginary parts, $ u_b = u_{b\,\rr} + \ii u_{b\,\ri}$ and $ \psi = \psi_{\rr} + \ii \psi_{\ri}$, the real part of (\ref{4eq:loopMKdV2}) is reduced to the gauged MKdV equation with gauge field $A(u)=(\lambda_{2g}+3b_a-\frac{3}{4}(\partial_{u_{g}\, \rr}\psi_\ri)^2)/2$,
\begin{equation}
-(\partial_{u_{g-1}\, \rr}-
A(u)\partial_{u_{g}\, \rr})\psi_\rr
           +\frac{1}{8}
\left(\partial_{u_g\, \rr} \psi_\rr\right)^3
+\frac{1}{4}\partial_{u_g\, \rr}^3 \psi_\rr=0
\label{4eq:gaugedMKdV2}
\end{equation}
by the Cauchy-Riemann relations as mentioned in \cite[(11)]{MP22}.
The conserved energy $\displaystyle{\frac12\int (\partial_s \psi)^2 ds =}$
$\displaystyle{\frac12\int [(\partial_s \psi_\rr)^2-(\partial_s \psi_\ri)^2] ds}$
$+\ii\displaystyle{\int [(\partial_s \psi_\rr)(\partial_s \psi_\ri)] ds}$.

\bigskip
In order to obtain a solution of (\ref{4eq:rMKdV2}) or excited states of elastica
in terms of the data in Theorem \ref{4th:MKdVloop},  the following conditions must be satisfied \cite{MP22}:

\begin{enumerate}

\item[CI] $\prod_{i=1}^g |x_i - b_a|=$ a constant $(> 0)$ for all $i$ in Theorem \ref{4th:MKdVloop},

\item[CII] $d u_{g\,\ri}=d u_{g-1\, \ri}=0$ in Theorem \ref{4th:MKdVloop}, and

\item[CIII] $A(u)$ is a real constant:
if $A(u)=$ constant, (\ref{4eq:gaugedMKdV2}) is reduced to (\ref{4eq:rMKdV2}). 
\end{enumerate}

\section{Hyperelliptic Curves of Genus Three}\label{sec:g=3}

For the reality conditions we consider the hyperelliptic curves $X$ of genus three.
According to the associated paper \cite{M24a}, we show the results of the real solution of the gauged MKdV equation and the MKdV equation in \cite{M24a} without proofs.

Let $\gamma \ee^{2\ii\varphi} :=(x-b_0)$ in (\ref{4eq:hypC}) of $g=3$.
We rewrite (\ref{4eq:hypC}) for the angle expression,
\begin{gather}
X:=\{(x, y) \ |\ y^2=-64 \frac{\gamma^4\ee^{8\ii\varphi}}{k_1^2 k_2^2k_3^2} 
(1-k_1^2 \sin^2 \varphi)(1-k_2^2 \sin^2 \varphi)
(1-k_3^2 \sin^2 \varphi)\},
\label{4eq:HEcurve_phi}
\end{gather}
where 
$\displaystyle{
k_a = \frac{2\ii\sqrt[4]{e_{2a-1}e_{2a}}}{\sqrt{e_{2a-1}}- \sqrt{e_{2a}}}
=\frac{\sqrt{\gamma}}{\beta_a}}$, $(a=1,2,3)$.

We note that (\ref{4eq:HEcurve_phi}) means the double covering $\hvarpi: \hX \to X$ and there exists the double covering Jacobi variety $\hvarpi_J: \hJ_X \to J_X$ \cite{M24a, M24}.
We have the holomorphic one-form in the covering Jacobi variety $\hJ_X$ for $S^g \hX$.
\begin{lemma} \label{4lm:dudphi}
For $(\gamma \ee^{\ii \varphi_i}, K_j)_{j=1, 2, 3}\in S^3\hX$, where $K_j:=\tgamma_j \tK(\varphi_j)$, ($j=1, 2, 3$), and $$\displaystyle{
\tK(\varphi):=\frac{\sqrt{\gamma
 (1-k_1^2 \sin^2 \varphi)(1-k_2^2 \sin^2 \varphi)
(1-k_3^2 \sin^2 \varphi)}}
{k_1k_2k_3}}, \quad \tgamma = \pm 1,
$$
 the following holds:

\noindent
$$\displaystyle{
\begin{pmatrix} d u_1 \\ d u_2\\ du_3\end{pmatrix}
=-
\begin{pmatrix}
\frac{ \ee^{-2\ii\varphi_1}}{8\gamma^2 K_1}&
\frac{ \ee^{-2\ii\varphi_2}}{8\gamma^2 K_2}&
\frac{ \ee^{-2\ii\varphi_3}}{8\gamma^2 K_3}\\
\frac{\ii \ee^{-\ii\varphi_1}\sin(\varphi_1)}{4\gamma K_1}&
\frac{\ii \ee^{-\ii\varphi_2}\sin(\varphi_2)}{4\gamma K_2}&
\frac{\ii \ee^{-\ii\varphi_3}\sin(\varphi_3)}{4\gamma K_3}\\
\frac{- \sin^2(\varphi_1)}{2 K_1}&
\frac{- \sin^2(\varphi_2)}{2 K_2}&
\frac{- \sin^2(\varphi_3)}{2 K_3}\\
\end{pmatrix}
\begin{pmatrix} d \varphi_1 \\ d \varphi_2 \\d \varphi_3\end{pmatrix}
}.
$$
Let the matrix be denoted by $\cL$.
Then $\det(\cL)=\displaystyle{
\frac{\sin(2\phi_2 - 2\phi_1) + \sin(2\phi_3 - 2\phi_2) + \sin(2\phi_1 - 2\phi_3)}{4^4 K_1 K_2 K_3}}$.
\end{lemma}

As we are concerned with the situation that $y/\ee^{4\ii \varphi}$ is real or pure imaginary, we assume that the branch points surround the circle whose center is $(b_0,0)$ and radius is $\gamma=1$.
Noting that we handle the double covering $\hX$ of $X$ with twelve branch points, we define $\varphi_{\fb a}^{+\pm}:=\pm \sin^{-1}(1/k_a)$ and $\varphi_{\fb a}^{-\pm}
=\pi - \varphi_{\fb a}^{+\pm}$ as in Figure \ref{fg:Fig01}.

Further, for the case Figure \ref{fg:Fig01} (a) $k_1> k_2 > k_3>1.0$, we assume $\varphi_\fb^\pm:=\varphi_{\fb1}^{+\pm}$ whereas for the case Figure \ref{fg:Fig01} (b) we let $\varphi_\fb^\pm:=\varphi_{\fb1}^{\pm+}$.

\begin{figure}
\begin{center}

\includegraphics[width=0.41\hsize]{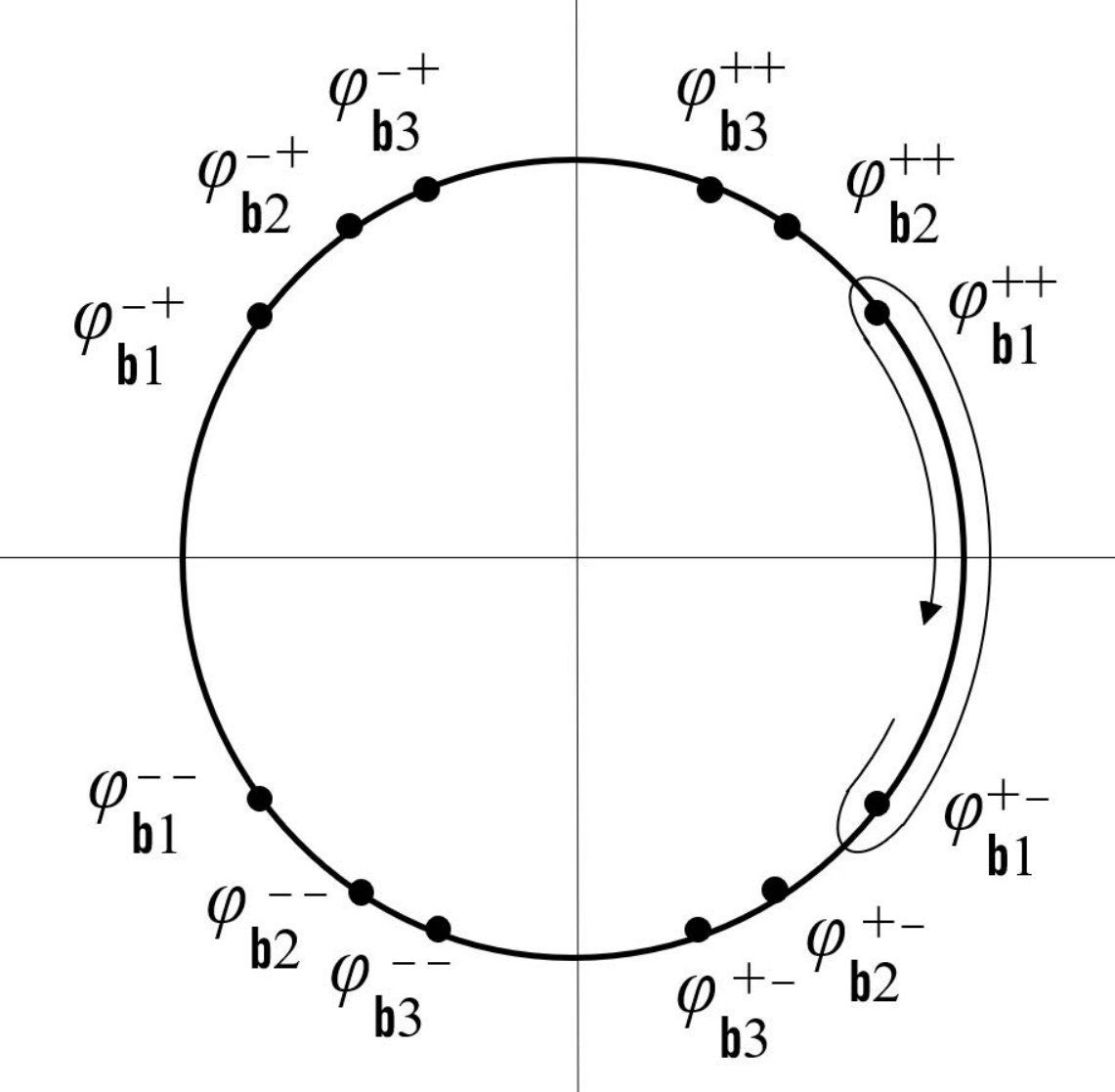}
\hskip 0.1\hsize
\includegraphics[width=0.42\hsize]{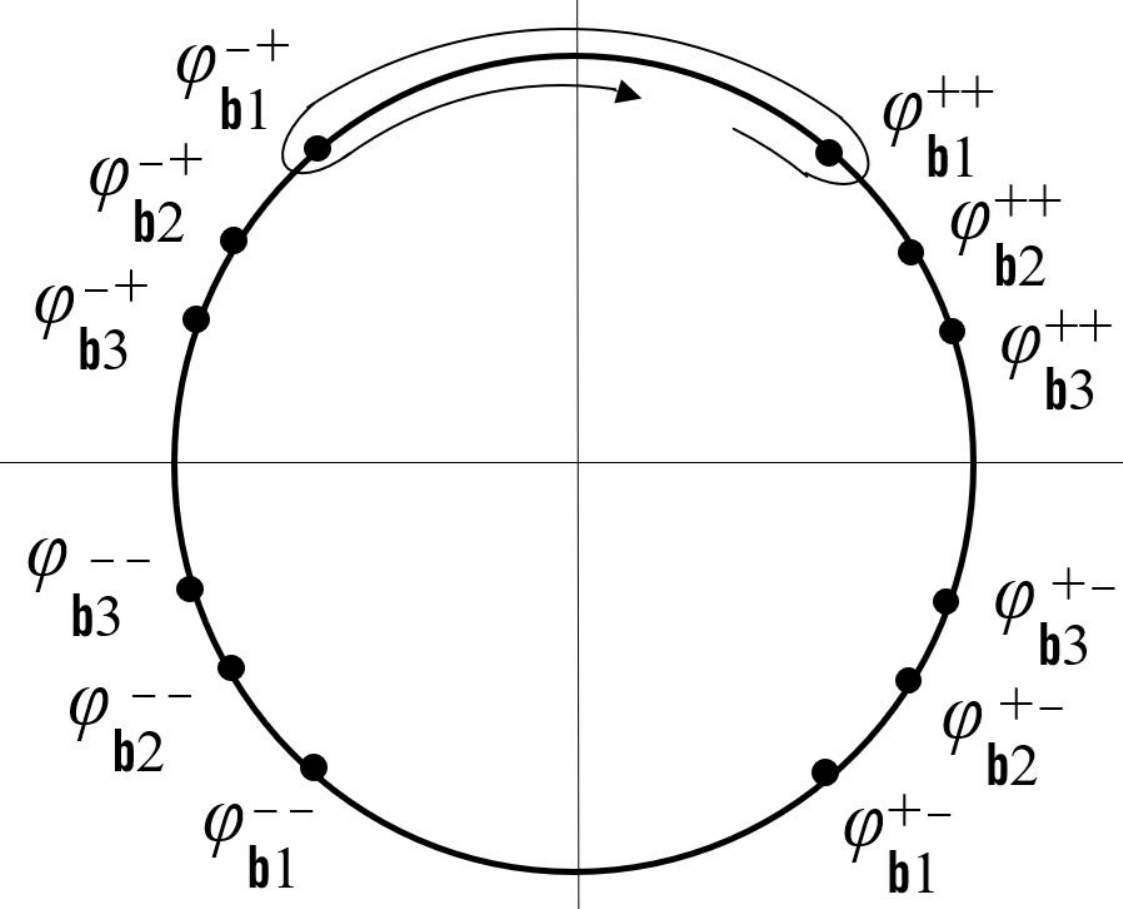}

(a) \hskip 0.35\hsize (b)

\end{center}

\caption{
The orbits of each $\varphi_i$ in the quadrature:
(a): $k_1> k_2 > k_3>1.0$.
(b): $k_3>k_2>k_1>1.0$.
}\label{fg:Fig01}
\end{figure}

We also have the inverse of Lemma \ref{4lm:dudphi}:

\begin{lemma} \label{4lm:dudphi3}
For $\varphi_j \in [\varphi_\fb^-, \varphi_\fb^+]$, $(j=1,2,3)$ such that $\varphi_\neq \varphi_j$ $(i\neq j)$, we have
\begin{equation}
\displaystyle{
\begin{pmatrix} d \varphi_1 \\ d \varphi_2 \\ d \varphi_3\end{pmatrix}
=\cK \cM
\begin{pmatrix} d u_1 \\ d u_2\\ du_3\end{pmatrix}
},
\qquad
\cL^{-1}=\cK \cM,
\label{4eq:Elas3.4}
\end{equation}
where
$
\displaystyle{
\cK
:=-
\begin{pmatrix}
\frac{K_1}{\sin(\varphi_2-\varphi_1)\sin(\varphi_3-\varphi_1)}&0& 0 \\
0& \frac{K_2}{\sin(\varphi_3-\varphi_2)\sin(\varphi_1-\varphi_2)}&0 \\
0&0&\frac{K_3}{\sin(\varphi_1-\varphi_3)\sin(\varphi_2-\varphi_3)} 
\end{pmatrix}
}
$ and,
{\small{
$$
\displaystyle{
\cM
:=
\begin{pmatrix}
8 \gamma^2 \sin\varphi_2\sin\varphi_3&
-4\ii\gamma(2\ii \sin\varphi_2\sin\varphi_3 - \sin(\varphi_2+\varphi_3) )&
-2 \ee^{-\ii (\varphi_2+\varphi_3)} \\
8  \gamma^2\sin\varphi_1\sin\varphi_3&
-4\ii\gamma(2\ii \sin\varphi_1\sin\varphi_3 - \sin(\varphi_3+\varphi_1) )&
-2 \ee^{-\ii (\varphi_1+\varphi_3)} \\
8 \gamma^2 \sin\varphi_1\sin\varphi_2&
-4\ii\gamma(2\ii \sin\varphi_1\sin\varphi_2 - \sin(\varphi_1+\varphi_2) )&
-2\ee^{-\ii (\varphi_1+\varphi_2)} \\
\end{pmatrix}.
}
$$
}}
\end{lemma}

\begin{proof}
The straightforward computations show it. \qed
\end{proof}

Let us focus on the differential equation,
\begin{equation}
\begin{pmatrix} d \varphi_{1,\rr} \\ d \varphi_{2,\rr} \\
d \varphi_{3,\rr}\end{pmatrix}
=\begin{pmatrix}
\frac{2K_1\tgamma_1\cos(\varphi_2-\varphi_3)}{\sin(\varphi_2-\varphi_1)\sin(\varphi_3-\varphi_1)}\\
\frac{2K_2\tgamma_2\cos(\varphi_3-\varphi_1)}{\sin(\varphi_3-\varphi_2)\sin(\varphi_1-\varphi_2)}\\
\frac{2K_3\tgamma_3\cos(\varphi_1-\varphi_2)}{\sin(\varphi_1-\varphi_3)\sin(\varphi_2-\varphi_3)}
\end{pmatrix}ds
\label{eq:g3CIII}
\end{equation}
where $ds$ is the one-form on the real axis, or $s \in \RR$.
Then we have
$\begin{pmatrix} 0 \\ -ds/2\\ ds\end{pmatrix}
=\cL \begin{pmatrix} d \varphi_{1,\rr} \\ d \varphi_{2,\rr} \\
d \varphi_{3,\rr}\end{pmatrix}$ 
Here $s$ corresponds to the arclength of the excited states of elastica and $u_{g,\rr}$ in (\ref{4eq:gaugedMKdV2}) of $g=3$.

Theorem 4.4 in \cite{M24a} reads as follows.
\begin{theorem}\label{pr:solgMKdV}
For a solution of the differential equation (\ref{eq:g3CIII}), $(\varphi_{1}(s), \varphi_{2}(s), \varphi_{3}(s)) \in \RR^3$ if exists, we let $\psi_{\rr}(s)=2(\varphi_{1}(s)+ \varphi_{2}(s)+\varphi_{3}(s))$, 
and then $\psi_\rr(s)$ is a cross-section of $\psi_\rr(u_2, u_3)|_{u_3=s, u_2=-s/2}$ of $\psi_\rr(u_2, u_3)$ which satisfies the gauged MKdV equation,
\begin{equation}
(\partial_{u_1}+\tA(u_2,u_3)\partial_{u_2})\psi_\rr(u_2, u_3)
           +\frac{1}{8}
\left(\partial_{u_3} \psi_\rr(u_2, u_3)\right)^3
+\frac{1}{4}\partial_{u_3}^3 \psi_\rr(u_2, u_3)=0,
\label{4eq:gaugedMKdV2a}
\end{equation}
where the gauge field is $\tA(u_1, u_2)=(\lambda_{6}-3-\frac{3}{4}(\partial_{u_3}\psi_\ri(u_2, u_3))^2)/2$ given by $\partial_{u_3}\psi_\ri(u_2, u_3)=\partial_{u_3} \varphi_{1,\ri}+\partial_{u_3} \varphi_{2,\ri}+\partial_{u_3} \varphi_{3,\ri}$, and
\begin{equation}
\begin{pmatrix} 
\partial_{u_3} \varphi_{1,\ri} \\ 
\partial_{u_3} \varphi_{2,\ri} \\
\partial_{u_3} \varphi_{3,\ri}
\end{pmatrix}
=\begin{pmatrix}
\frac{2K_1\tgamma_1\sin(\varphi_2+\varphi_3)}{\sin(\varphi_2-\varphi_1)\sin(\varphi_3-\varphi_1)}\\
\frac{2K_2\tgamma_2\sin(\varphi_3+\varphi_1)}{\sin(\varphi_3-\varphi_2)\sin(\varphi_1-\varphi_2)}\\
\frac{2K_3\tgamma_3\sin(\varphi_1+\varphi_2)}{\sin(\varphi_1-\varphi_3)\sin(\varphi_2-\varphi_3)}
\end{pmatrix}.
\label{eq:g3CIIIi}
\end{equation}
\end{theorem}

We note that (\ref{4eq:gaugedMKdV2a}) is a differential identity of $\psi_{\rr}$ for every hyperelliptic curve as in \cite{Mat02b}.

It means the following theorem:
\begin{theorem}\label{4th:reality_g3}
$\psi_{\rr}:=2(\varphi_1+ \varphi_2+\varphi_3)$ of the quadrature $d \varphi_{i, \rr}$ $(i=1,2,3)$ of (\ref{eq:g3CIII}) shows a local solution of the MKdV equation (\ref{4eq:rMKdV2}) for the region that $\partial_{u_3} \psi_\ri$ is constant.

Further on the region that $\partial_{u_3} \psi_\ri$ is approximately constant, the $\psi_{\rr}$ shows the approximate solutions of the MKdV equation.

\end{theorem}

\section{Numerical results and discussion}\label{sec:Algorithm}

We demonstrate the shapes of excited states of elastica of genus three, a closed solution and three open solutions.

We show the numerical integration (\ref{eq:g3CIII}) as follows.
\begin{enumerate}

\item We set $(k_1, k_2, k_3)$ to determine a hyperelliptic curve.
We assume Figure \ref{fg:Fig01} (a) or (b).
$k_1> k_2 > k_3>1.0$ and $k_3>k_2>k_1>1.0$.
We explain mainly the case (a).
Let $\varphi_{\fb}:=\varphi_{\fb1}^{++}$.

\item We set the initial condition $(\varphi_1, \varphi_2, \varphi_3)|_{s=0}$ that the imaginary part of $2\partial(\varphi_{1,\ri}+ \varphi_{2,\ri}+ \varphi_{3,\ri})$ in (\ref{eq:g3CIIIi}) vanishes.

\item We employ the Euler method of the numerical quadrature method for a sufficiently small real value $\delta s$ to solve (\ref{eq:g3CIII}) numerically so that the orbit of $\varphi_a$ moves back and forth between the branch points $[-\varphi_\fb, \varphi_\fb]$ in (\ref{4eq:HEcurve_phi}) as in Figure \ref{fg:Fig01} (a).

We note that at the branch point, the orbit of $\varphi_a$ turns the direction by changing the sign of $\tgamma_a$ so that $\tgamma_a \sin^2(\varphi_a)\delta \varphi_a/2K_a$ is positive as in Figure \ref{fg:Fig01};
it moves the different leaf of the Riemann surface with respect to the projection $X_3 \to \PP$ $((x,y) \mapsto x)$ after passing the branch points.

Following the Euler method of the numerical quadrature method for the $n$-step, we obtain the $\varphi_{a}$ development,
$$
\varphi_{a,n+1} := \varphi_{a,n} + \delta \varphi, \quad
(a=1,2,3).
$$
We let $\psi_n = 2(\varphi_{1,n}+\varphi_{2,n}+\varphi_{3,n}) + \psi_c$ so that $\psi_0$ is a certain value (in the following results, $\psi_0=0$), and numerically integrate
$$
X_{n+1} = X_n + \cos(\psi_n) \delta s, \quad
Y_{n+1} = Y_n + \sin(\psi_n) \delta s,
$$
to obtain the excited states of elastica $(X(s), Y(s))$ of a certain $(u_1, u_2, u_3)$ point in $J_X$.

\item We monitor the imaginary part of $\psi_\ri:=\displaystyle{\int^s \partial_{u_3} \psi_{\ri}ds}$ for $\partial_{u_3} \psi_\ri$ $=2\partial(\varphi_{1,\ri}+ \varphi_{2,\ri}+ \varphi_{3,\ri})$ in (\ref{eq:g3CIIIi}).
Though the above theorem can hold if $\partial_{u_3} \psi_\ri=$constant number $c\in \RR$, we focus on the case where $c\sim0$.
Thus we monitor the flatness of $\psi_{\ri}$.

\end{enumerate}

The first result is displayed in Figure \ref{fg:shape01}.
For the hyperelliptic curve given by $(k_1, k_2, k_3) = (1.04, 1.0392, 1.010)$, we set the initial condition $(\varphi_1, \varphi_2, \varphi_3) = (\varphi_\fb, -0.90, -0.90)$.
Figure \ref{fg:shape01} (a) shows the profile of $\psi_\rr$ and $\psi_\ri$ and (b) shows the shape of the excited states of elastica.
The maximum of $\partial_{u_3} \psi_\ri$ is $1.00952\times 10^{-02}$.
In other words, the orbit is considered to be one of the MKdV equation (\ref{4eq:rMKdV2}) rather than the gauged MKdV equation (\ref{4eq:gaugedMKdV2}).

\begin{figure}
\begin{center}
\includegraphics[width=0.50\hsize]{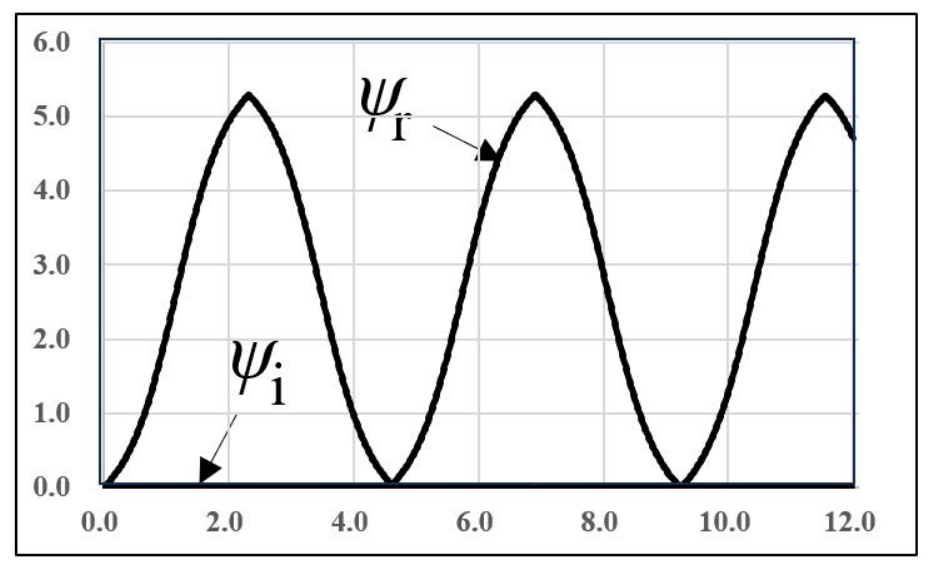}
\hskip 0.1\hsize
\includegraphics[width=0.30\hsize]{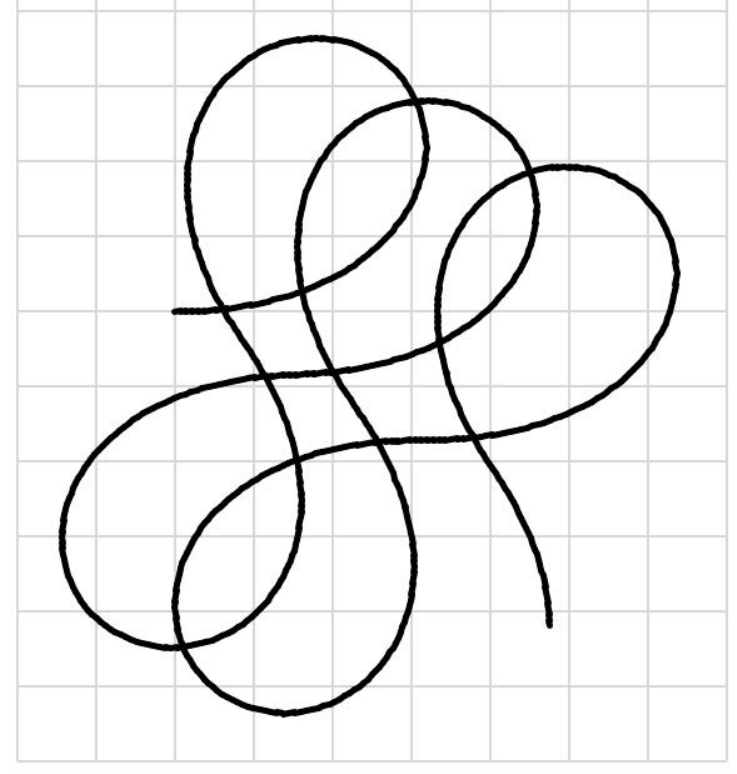}

(a) \hskip 0.35\hsize (b)

\end{center}

\caption{
An open excited state of elastica:
$(k_1, k_2, k_3) = (1.04, 1.0392, 1.010)$, and the initial condition is  $(\varphi_1, \varphi_2, \varphi_3) = (\varphi_\fb, -0.90, -0.90)$.
(a): $\psi_\rr$ and $\psi_\ri$, and 
(b): a shape of the excited state of elastica. 
}\label{fg:shape01}
\end{figure}

The second result is presented in Figure \ref{fg:shape02}.
We used the hyperelliptic curve given by $(k_1, k_2, k_3) = (1.0260, 1.0259, 1.0008)$.
The initial condition is given by $(\varphi_1, \varphi_2, \varphi_3) = (\varphi_\fb, 0.0, 0.0)$.
Figure \ref{fg:shape02} (a) shows the profile of $\psi_\rr$ and $\psi_\ri$ and (b) shows the shape of the excited state of elastica.
The orbit could be considered as an approximation of the MKdV equation (\ref{4eq:rMKdV2}) rather than the measured MKdV equation (\ref{4eq:gaugedMKdV2}).
Figures \ref{fg:shape02} (e) and (f) show the respective parts of $\partial_{u_3} \varphi_{i,\rr}$ and $\partial_{u_3} \varphi_{i,\ri}$.

\begin{figure}
\begin{center}
\includegraphics[width=0.38\hsize]{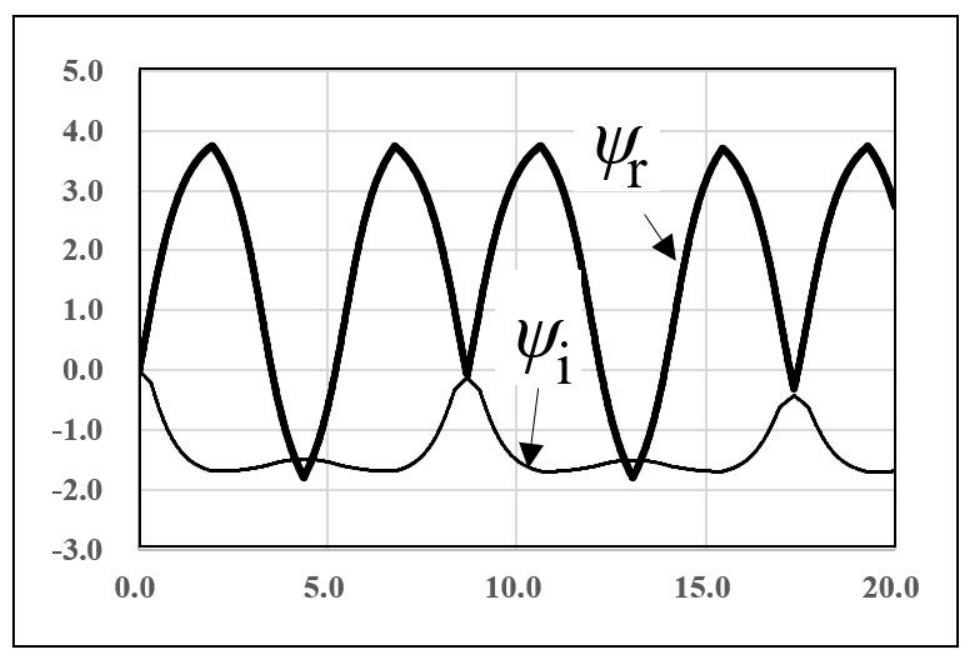}
\hskip 0.05\hsize
\includegraphics[width=0.53\hsize]{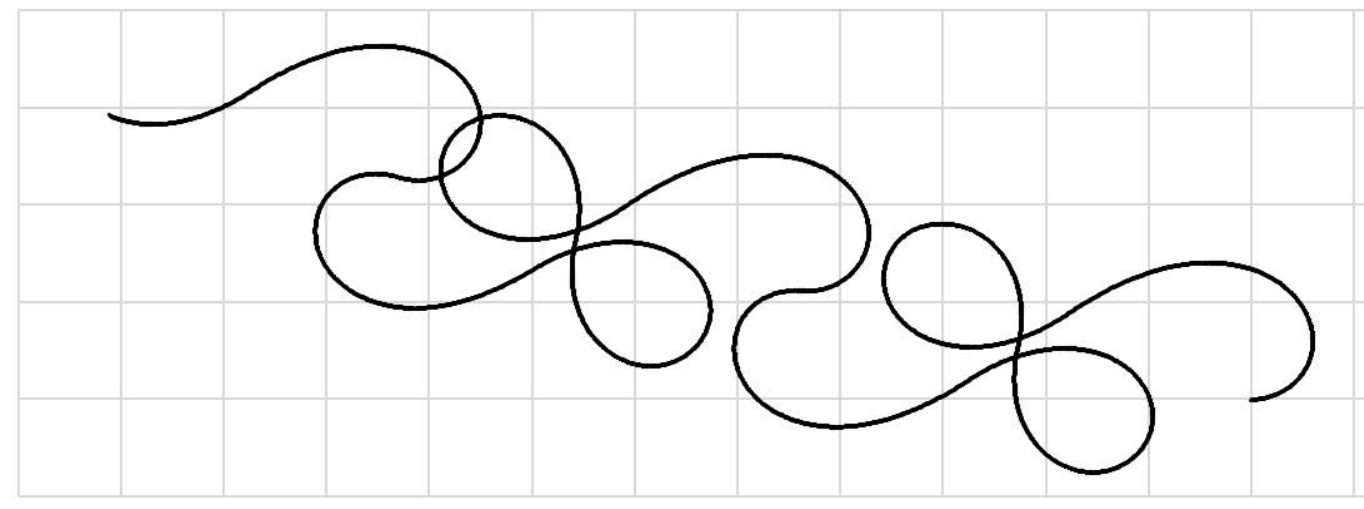}

(a) \hskip 0.3\hsize (b)

\end{center}

\caption{
An open excited state of elastica:
$(k_1, k_2, k_3) = (1.0260, 1.0259, 1.0008)$ 
$(\varphi_1, \varphi_2, \varphi_3) = (\varphi_\fb, 0.0, 0.0)$
(a) is the profile of $\psi_\rr$ and $\psi_\ri$ and 
(b) is the shape of the excited state of elastica.
}\label{fg:shape02}
\end{figure}

\begin{figure}
\begin{center}
\includegraphics[width=0.43\hsize]{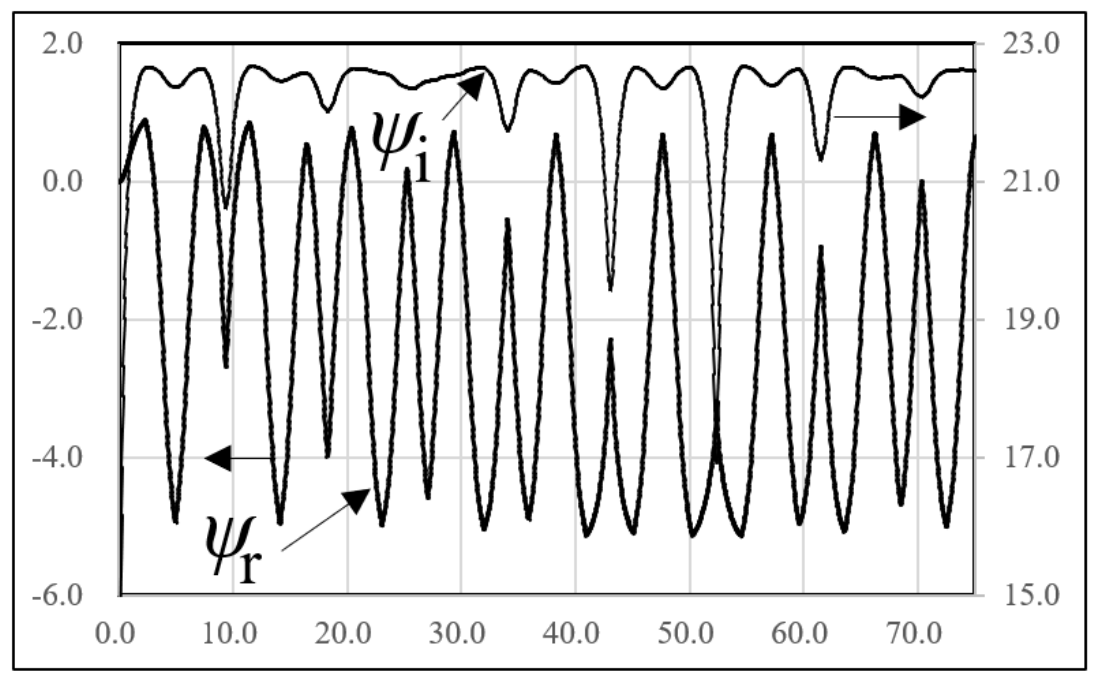}
\hskip 0.00\hsize
\includegraphics[width=0.55\hsize]{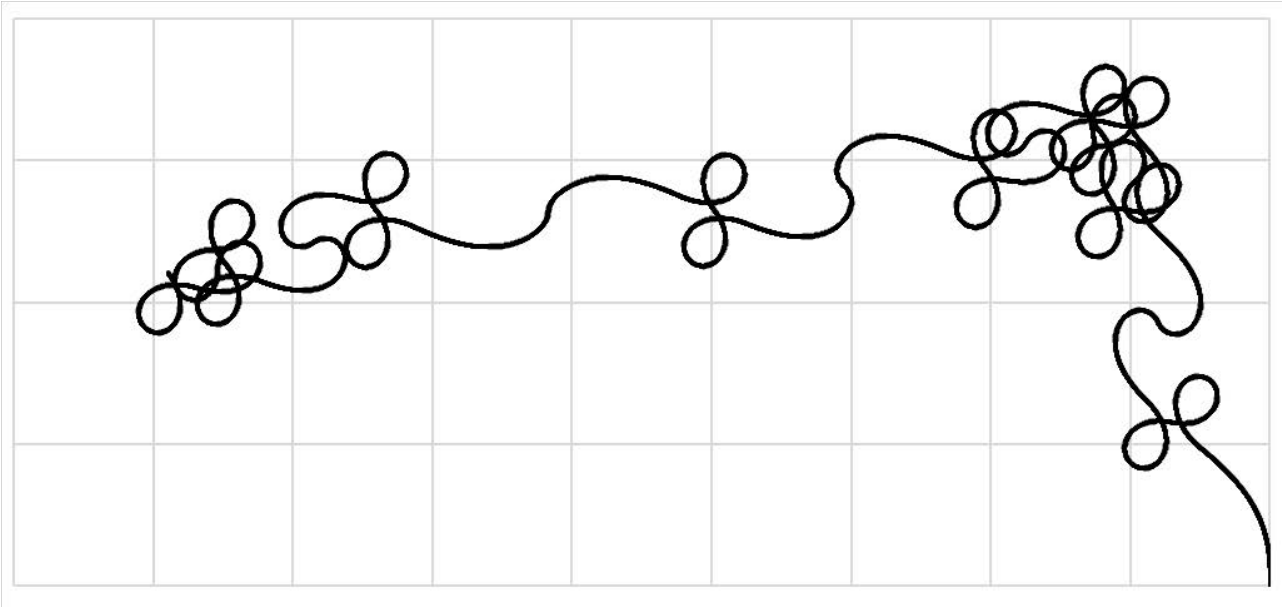}

(a) \hskip 0.35\hsize (b)

\smallskip


\includegraphics[width=0.55\hsize]{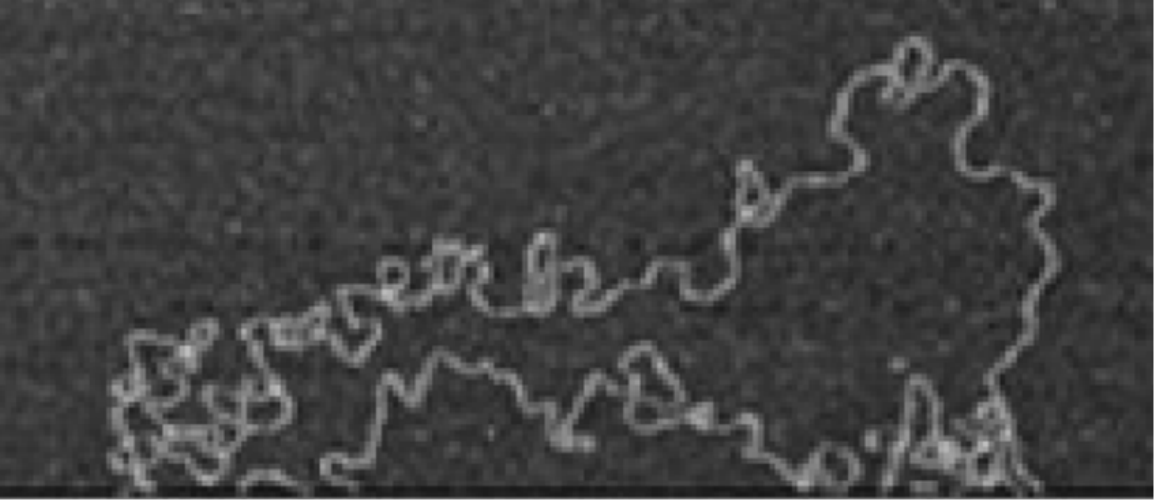}

(c)

\end{center}

\caption{
An open curve of excited state of elastica and shape of supercoiled DNA:
$(k_1, k_2, k_3) = (6.00, 7.00, 8.00)$, and the initial condition is  
$(\varphi_1, \varphi_2, \varphi_3) = (\pi-\varphi_\fb, 1.4, 1.4)$.
(a) is the profile of $\psi_\rr$ and $\psi_\ri$, 
(b) is a shape of the excited state of elastica, and 
(c) is the shape of a supercoiled DNA, which is a part of the AFM images in \cite[Figure 4]{JMB}.
}\label{fg:shape04}
\end{figure}

\bigskip

\bigskip

Similarly, we can use $\ii ds$ instead of $ds$ and consider the case in Figure \ref{fg:Fig01} (b).
Since replacing $ds$ with $\ii ds$ does not seriously affect the MKdV equation (\ref{4eq:rMKdV2}) by replacing $(dt, \alpha)$ with $(-\ii dt, -\alpha)$, we can use the same algorithm to compute the shape of the excited state of elastica by considering $\ii K_a$ instead of $K_a$ for Figure \ref{fg:Fig01} (b).
The result is Figure \ref{fg:shape04}.
Figure \ref{fg:shape04} (a) and (b) show profile of $\psi_\rr$ and $\psi_\ri$, and the shape of the excited state of elastica with parameters $(k_1, k_2, k_3) = (6.00, 7.00, 8.00)$ and initial condition $(\varphi_1, \varphi_2, \varphi_3) = (\pi-\varphi_\fb, 1.4, 1.4)$.
Due to $\psi_{\rr}$ and $\psi_{\ri}$ in Figure \ref{fg:shape04} (a), when $\partial_{u_3}\psi_{\ri}$ is constant, $\psi_{\rr}$ is a solution of the MKdV equation (\ref{4eq:rMKdV2}). 
The third example also consists of figure-eight and the inverse of the figure {\lq}S{\rq}, a repetition of the modulation of the figure-eight and the inverse {\lq}S{\rq}, as in Figure \ref{fg:shape04} (b).
We call it the S-eight mode.

In other words, we conclude that our model of the statistical mechanics of elastica shows that in the excited states of elastica due to thermal effect, we find the S-eight mode.
This fact has never been reported anywhere.

On the other hand, in the shape of the supercoiled DNAs in Figure \ref{fg:shape04}(c) \cite[Figure 4]{JMB}, it is surprising that we can find the similar shape, a repetition of the modulation of the figure-eight and the inverse {\lq}S{\rq}, or the S-eight mode.
In their study of a new type of highly ordered DNA organization, Japaridze et al.\~investigated the conformation of DNA both experimentally and theoretically.
They found the new type of order called "hyperplectonemes".
The shape in the figure is shown as the case that has fewer hyperplectonemes. 
In fact, in the case there are similar shapes whose parts have figure-eight and inverse S shapes, or the S-eight mode.
After the study, some of the authors found a parameter which governs the order and controls the degree of order in \cite{OLH20}.
Then we also find the shapes similar to Figure \ref{fg:shape04} (c) in \cite[Figure 7 (c)]{OLH20}, whose part is expressed by the modulation of figure-eight and inverse {\lq}S{\rq}, as in Figure \ref{fg:shape04} (a).

We emphasize that except for the figure-eight given by Euler in 1744, which is found as short closed supercoiled DNAs, e.g. in \cite{Petal}, no one has ever mathematically reproduced any shape of supercoiled DNA with voids.
Thus, this demonstration provides the first step towards the mathematical representation of the conformations of supercoiled DNA.
In other words, the excited states of elastica of genus three of the statistical mechanics of elastica reproduces the geometric property of supercoiled DNA.
The shapes of supercoiled DNA are not rigid, but generally obey weak elastic forces.

Since the MKdV equation preserves the Euler-Bernoulli energy $\displaystyle{\int (\partial_s \psi)^2 ds}$, it can be considered an excited state of elasticity rather than a ground state or minimum energy point.
This means that we are beginning to go beyond the form of Euler's elastica, including thermal effects. \cite{Mat10}.

Thus, we conclude our model of the statistical mechanics of elastica to express a certain class of supercoiled DNA shapes observed in the laboratory.
The S-eight mode exits in the class of the shapes of the supercoiled DNAs.
It shows that there is a deep and beautiful relationship between nature (life sciences) and algebraic geometry in mathematics.

\section{Conclusion}

In this paper, we employed a novel algebro-geometric method in \cite{M24a} to obtain the solutions of the excited states of elastica of the gauged MKdV equation (\ref{4eq:gaugedMKdV2}) as in Proposition \ref{pr:solgMKdV}.
Theorem \ref{4th:reality_g3} shows that they can also be regarded as approximate solutions of the MKdV equation (\ref{4eq:rMKdV2}) if $\partial_{u_3} \psi_{\ri}$ is approximately constant.
Based on them, we provided a concrete algorithm to obtain the numerical solutions of the excited states of elastica in Section \ref{sec:Algorithm}.
We demonstrate typical conformations of the excited states of elastica by numerical computations, which no one ever draw such a complicated shape.

We found that there is a mode, a repetition of the modulation of the figure-eight and the inverse {\lq}S{\rq}.
We call it the S-eight mode.
It means that our model of the statistical mechanics of elastica shows that the excited states of elastica due to thermal effect has the S-eight mode.
Then we found that there exist the shapes with the S-eight mode in the AFM image of the supercoiled DNAs in  \cite[Figure 4]{JMB}.
We conclude that our model expresses a certain class of the shape of supercoiled DNA.
Only the shape of the supercoiled DNA related to figure-eight and circle of Euler's has been expressed by theoretically and numerically.
Though the more complicated ones with void has never obtained, this paper shows such a shape by the algebro-geometric method.

There are so many mathematical models of algebraic geometry related to elementary particle physics and string theory.
This paper shows that there is a deep and beautiful relationship between the biophysics of DNAs and algebraic geometry as mentioned in \cite{Mat10}.

\bigskip
However, while the results of this paper are certainly novel and intriguing, a number of issues need to be resolved in the future.
Though we have computed it, the behavior of the imaginary part $\partial_{u_3} \psi_{\ri}$ is unclear.
The behavior should be considered more precisely in the future to find the solutions of the MKdV equation.
Based on the knowledge, we should find the shapes of the excited states of elastica with the higher genus $g>3$.

Furthermore, we should find the hyperelliptic solutions of the NLS equation beyond \cite{P0} to obtain the excited states of elastica in $\RR^3$ as in \cite{Mat99a}.

\bigskip

\noindent
{\bf{Acknowledgment}:}
This project was started with Emma Previato 2004 in Montreal and had been collaborated until she passed away June 29, 2022.
Though the author started to step to genus three curves without her, he appreciate her contributions and suggestions which she gave him to this project during her lifetime.
Further, it is acknowledged that John McKay who passed way April 2022 invited the author and her to his private seminar in Montreal 2004 since he considered that this project \cite{Mat97} must have been related to his Monster group problem \cite{McKay, MP16}.
This problem might be related to Witten conjecture of loop space via the sigma function \cite{MP16}.
Thus this study is devoted to Emma Previato and John McKay.
The author thanks to Junkichi Satsuma, Takashi Tsuboi, and Tetsuji Tokihiro for inviting him to the Musashino Center of Mathematical Engineering Seminar and for valuable discussions and to Yuta Ogata, Yutaro Kabata and Kaname Matsue for helpful discussions and suggestions.
He is also grateful to Aleksandre Japaridze, Giovanni Longo, and Giovanni Dietler, the authors of \cite{JMB} for helpful comments on Figure 4 in \cite{JMB} and sending him its follow-up interesting article \cite{OLH20}.
He also acknowledges support from the Grant-in-Aid for Scientific Research (C) of Japan Society for the Promotion of Science, Grant No.21K03289.

%
%

\begin{thebibliography}{00}

\bibitem{AS}
\by{M. J. Ablowitz, H. Segur}
\book{Solitons and the inverse scattering transform}
SIAM 1981.


\bibitem{AdamsHarnadPreviato}
\by{M.~R.~Adams, J.~Harnad, and E.~Previato}
\paper{Isospectal Hamiltonian Flows in Finite and Infinite Dimensions}
\jour{Comm.~Math.~Phys.}   \yr{1988} \vol{117} 451-500.


\bibitem{Baker97}
  \by{H.F. Baker}
  \book {Abelian functions}
  \publ{Cambridge Univ. Press, Cambridge, 1995}
   Reprint of the 1897 original.



\bibitem{BM}
C. Bouchiat and M. M\'ezard,
\textit{Elastic rod model of a supercoiled DNA molecule},
Eur. Phys. J. E {\bf{2}} (2000) 377-402.

\bibitem{Brylinski} 
\by{J.~-L.~Brylinski}
\book{Loop spaces, characteristic classes and geometric quantization,
 Progress in Mathematics 107}
\publ{Birkh\"auser}
\publaddr{Boston}
 1993.  


\bibitem{Betal}
T. Brouns et al.,
\textit{
Free Energy Landscape and Dynamics of Supercoiled DNA by High-Speed Atomic 
Force Microscopy}, 
ACS Nano. 2018 Dec 26;12(12):11907-11916. 

\bibitem{BEL97b}
\by{V.M. Buchstaber,  V.Z. Enolski\u{\i} and D.V. Le\u{\i}kin}
\paper{Kleinian functions, hyperelliptic Jacobians and applications}
\jour{Rev. Math. Math. Phys.}
\vol{10} (1997)  1--103.

\bibitem{BEL20}
\by{V.~M.~Buchstaber,
V.~Z.~Enolski, and 
D.~V.~Leykin}
\paper{$\sigma$ functions: old and new results}
in
\book{
Integrable systems and algebraic geometry vol.2} 
edited by R. Donagi and T. Shaska,
London Math Soc. Lect. Note Series \vol{459} (2020) 
\pages{175--214}.


\bibitem{BuL04}
\by{V.~M.~Buchstaber and D.~V.~Leykin} 
\paper{Heat equations in a nonholonomic frame}
\jour{Funct.~Anal.~Appl.}
\vol{38} 
\yr{2004}
\pages{88--101}.

\bibitem{CDLT}
C. R. Calladine, H. Drew, B. Luisi and A. Travers,
Understanding DNA, 3rd ed.,
Academic Press, 2004.


\bibitem{DoiEdwards}
\by{M. Doi S. F. Edwards}
\book{The Theory of Polymer Dynamics (International Series of Monographs on Physics)}
Clarendon Press, 1988.


\bibitem{Euler44}
\by{L.~Euler} \book{Methodus Inveniendi Lineas Curvas Maximi
Minimive Proprietate Gaudentes} 1744.

\bibitem{EEMOP08}
\by{J.~C.~Eilbeck, V.~Z.~Enol'skii, S.~Matsutani, 
Y.~\^Onishi, and E.~Previato}
\paper{Addition formulae over the Jacobian pre-image of hyperelliptic 
Wirtinger varieties} 
\jour{J.~Reine Angew Math.~} 
\vol{619}  \yr{2008} \pages{37--48}.



\bibitem{EHKKLS}
\by{V. Enolski, B. Hartmann, V. Kagramanova, J. Kunz, C. L\"ammerzahl, P. Sirimachan}
\paper{ Inversion of a general hyperelliptic integral and particle motion in Ho?ava?Lifshitz black hole space-times}
\jour{J. Math. Phys.}  \vol{53} \yr{2012} 012504.

\bibitem{EMO08}
  \by{V.~Z.~Enolskii, S.~Matsutani and Y.~\^Onishi}
  \paper{The addition law attached to a stratification
for a hyperelliptic Jacobian variety}
\jour{Tokyo J.~Math.} 
\vol{31} \yr{2008} \pages{27--38}.

\bibitem{FarkasKra}
\by{H.~M.~Farkas and I.~Kra}
\book{Riemann Surfaces (GTM 71)}
\publ{Springer-Verlag} \publaddr{New York} 1991.


\bibitem{Frenkel}
\by{D. Frenkel}
\paper{ Statistical Mechanics of Liquid Crystals}
 in
\book{Liquids, Freezing and the Glass Transition, Volume Part II (Les Houches) }
ed. by J.~-P.~Hansen, D.~Levesque and J.~Zinn-Justin,
North-Holland, 1991. pp. 691-762 

\bibitem{GoldsteinPetrich1}  
\by{R.~E.~Goldstein and D.~M.~Petrich}
\paper{The Korteweg-de Vries
hierarchy as dynamics of closed curves in the plane} 
\jour{Phys.~Rev.~Lett.} 
\vol{67} 
\yr{1991}
\pages{3203--3206}.


\bibitem{GPL} 
S. Goyal et al.,
\textit{Nonlinear dynamics and loop formation in Kirchhoff rods
with implications to the mechanics of DNA and cables},
J. Comp. Phys., {\bf{209}} (2005) 371.


\bibitem{Itzykson} 
\by{C. Itzykson, J.-M. Drouffe}
\book{Statistical Field Theory: Vol. 2}
\publ{Cambridge Univ.} 1989.



\bibitem{JMB} 
 A. Japaridze, et al.,
\paper{Hyperplectonemes: A Higher Order Compact and Dynamic DNA Self-Organization}
\jour{Nano Lett.} \vol{17} 3, \yr{2017} 1938--1948

\bibitem{KMP22}
\by{J. Komeda, S. Matsutani, E. Previato},
\paper{Algebraic construction of the sigma function for 
general Weierstrass curves}
\jour{Mathematics (MDPI) } \vol{10} (16) \yr{2022} \pages{10, 32}.

\bibitem{KP}
E. A. K\"ummerle, and E. Pomplun,
\textit{A computer-generated supercoiled model of the pUC19 plasmid}, 
Eur Biophys J., {\bf{34}} (2005) 13.

\bibitem{LS}
Y.   L. Lyubchenko  and L.   S. Shlyakhtenko,
\textit{Visualization of supercoiled DNA with atomic force
microscopy in situ},
Proc. Natl. Acad. Sci. USA,
{\bf{94}} (1997) 496.

\bibitem{McKay}
\by{J.~McKay and Y-H.~He}
 \paper{Kashiwa lectures on new approaches to the Monster}
\jour{
Notices, Int. Cons. Chinese Math.} \vol{10}
\yr{2022}
\pages{71-88}.





\bibitem{Mat97}
\by{S.~Matsutani}
\paper{Statistical mechanics of elastica on a plane}
\jour{J. Phys. A: Math. \& Gen.}
\vol{31}  \yr{1998} \pages{2705}.

\bibitem{Mat99a}
\by{S.~Matsutani}
\paper{Statistical mechanics of non-stretching elastica 
in three dimensional space}
\jour{J. Geom. Phys.}
\vol{29} \yr{1999} \pages{243}.



\bibitem{Mat02b}
\by{S. Matsutani} 
       \paper{Hyperelliptic loop solitons with genus $g$: 
       investigation of a quantized elastica},
      \jour{J.~Geom.~Phys.} \vol{43} 
     \yr{2002} 146.


\bibitem{Mat10}
\by{S. Matsutani} 
\paper{Euler's Elastica and Beyond}
\jour{J. Geom. Symm. Phys} \vol{17} \yr{2010}
\pages{45--86}.


\bibitem{M24a}
\by{ S. Matsutani}
\book{On real hyperelliptic solutions of focusing
modified KdV equation} 
\paper{arXiv:2309.04904}


\bibitem{M24}
\by{ S. Matsutani}
\book{The Weierstrass sigma function in higher genus and applications to integrable equations} to appear as {\lq}Monographs in Mathematics{\rq} Springer 2024, (p.484).

\bibitem{MO03a}
\by{S.~Matsutani and Y.~\^Onishi}
         \paper{On the moduli of a quantized elastica in $\PP$
and KdV flows: study of hyperelliptic curves as an
extension of Euler's perspective of elastica I}
          \jour{Rev.~Math.~Phys.} \vol{15} \yr{2003} 
\pages{559--628}.

\bibitem{MP16}
\by{ S. Matsutani and E. Previato}
 \paper{From Euler's elastica to the mKdV hierarchy, through the Faber 
polynomials}
\jour{J. Math. Phys.} \vol{57} (2016) 081519.


\bibitem{MP22}
\by{ S. Matsutani and E. Previato}
 \paper{An algebro-geometric model for the shape of supercoiled DNA}
\jour{Physica D} \vol{430} (2022) 133073.





\bibitem{OLH20}
\by{N T. Odermatta, et.al.}
\paper{Structural and DNA binding properties of mycobacterial integration host
factor mIHF}
\jour{J.  Struct. Biology} \vol{209} (2020) 107434.


\bibitem{P0}
E. Previato,
{\textit{
Hyperelliptic quasi-periodic and soliton solutions
of the nonlinear Schr\"odinger equation}},
Duke Math. J., {\bf{52}} (1985) 329-377.



\bibitem{Pr93}
E. Previato,
\textit{
Geometry of the modified KdV equation},
in 
\textit{
LNP 424: Geometric and quantum aspects of 
integrable systems}
Ed. by G.~F.~Helminck, Springer 1993, 43-65.

\bibitem{Petal}
{A.  L.  B.  Pyne et al.}
\textit{
Base-pair resolution analysis of the effect of supercoiling on DNA flexibility and major groove recognition by triplex-forming oligonucleotides},
Nature Comm. {\bf{12}} (2021) 1053.

\bibitem{Ramond}
\by{P. Ramond}
\book{Field Theory: A Modern Primer, 2nd ed.}
\publ{Addison-Wesley} 1989.

\bibitem{Schwinger}
\by{J. Schwinger}
\paper{On Gauge Invariance and Vacuuwn Polarization}
\jour{Phys. Rev.} \vol{82} \yr{1951} \pages{664-678}


\bibitem{SCT}
D. Swigon, et.al.,
\textit{The Elastic Rod Model for DNA and Its Application to the Tertiary
Structure of DNA Minicircles in Mononucleosomes}
Biophysical J. {\bf{74}} (1998) 2515.





\bibitem{TsuruWadati}
\by{H. Tsuru and M. Wadati}
\textit{Elastic Model of Highly Supercoiled DNA},
Biopolymers, {\bf{25}} (1986) 2083.


\bibitem{VV}
D Voet and J. G. Voet 
Fundamentals of Biochemistry: Life at the Molecular Level, 4th ed.
Wiley 2015.

\bibitem{Witten}
\by{E. Witten}
\paper{Two-dimensional gravity and intersection theory on moduli space}
\jour{Surv. Diff. Geom.} \vol{1} \yr{1991} \pages{243--310}




\end{thebibliography}
%

\end{document}